\documentclass{article}

\usepackage{amsmath}
\usepackage{tikz}
\usetikzlibrary{arrows,automata}
\usepackage{algorithm}
\usepackage{algorithmic}
\usepackage{latexsym}
\usepackage{fullpage}

\newtheorem{defi}{Definition}
\newtheorem{lemm}{Lemma}
\newtheorem{coro}[lemm]{Corollary}
\newtheorem{theo}[lemm]{Theorem}
\newenvironment{proof}{\par\noindent{\bf Proof}}%
{\hspace{\stretch{1}}$\Box$}

\newcommand{\prefix}{\sqsubseteq}

\newcommand{\naturals}{\mathbf{N}}

\newcommand{\qi}[0]{q_{\mathrm{I}}}
\newcommand{\sigmao}[0]{\Sigma_{\mathrm{o}}}
\newcommand{\Sigmao}[0]{\sigmao}
\newcommand{\path}[0]{\mathit{\rho}}
\newcommand{\trans}[1]{\stackrel{#1}{\rightarrow}}
\newcommand{\lang}[0]{\mathcal{L}}
\newcommand{\langf}[1]{\lang_{#1}}
\newcommand{\obs}[0]{\mathrm{obs}}

\newcommand{\dmin}[1]{\mathit{dmin}_{#1}}
\newcommand{\dmax}[1]{\mathit{dmax}_{#1}}
\newcommand{\distances}{\mathit{distances}_F}
\newcommand{\belief}{\mathcal{B}}
\newcommand{\optimal}[1]{{#1}^\star}
\newcommand{\optipred}{\optimal{P}}

\newcommand{\set}[1]{\{#1\}}

\newcommand{\queue}[0]{\mathcal{Q}}
\newcommand{\nbsucc}[0]{\mathrm{nsucc}}

\begin{document}
\title{Interval Predictability \mbox{in Discrete Event Systems}}
\author{Alban Grastien%
\thanks{A.~Grastien is with NICTA, Australia, 
and the Australian National University.}
}

\maketitle

\begin{abstract}
  In this paper we study the problem of predictability 
  in partially observable discrete event systems, 
  i.e., the question whether an observer 
  can predict the occurrence of a fault.  
  We extend the definition of predictability 
  to consider the time interval where the fault will occur: 
  the $(i,j)$-predictability does not only specify 
  that the fault will be predicted before it occurs, 
  but also that the predictor will be able 
  to predict that its occurrence will occur 
  in $i$ to $j$ observations from now.  
  We also provide a quadratic algorithm 
  that decides predictability of the system.  

\end{abstract}

\section{Motivation}
A fault is predictable 
if its unavoidable occurrence can always be determined in advance.  
Being able to predict the fault 
allows the supervisor to step in 
and take preventive actions, 
such as reconfiguring the system, 
replacing damaged components, or shutting the system down.  

Predictability has been greatly studied in the last decade 
(some references are provided in the related work section).  
To be maximally effective, 
the prediction should satisfy two criteria: 
it should be made well in advance, 
so that the operator has enough time 
to decide for and perform corrective actions; 
it should be reasonably precise, 
so that the repair is not performed too early 
if that is unnecessary.  
The first contribution of this paper is the formalisation 
of these two objectives: 
we define the notion of $(i,j)$-predictability, 
a generalisation of the existing notion of predictability 
that states that faults can always be predicted 
at least $i$ timesteps in advance 
and, when this prediction is made, 
the fault will not occur in more than $j$ timesteps.  

We study this definition of predictability and 
we propose an algorithm that computes all pairs $(i,j)$ 
for which predictability holds.  
We show that this algorithm runs in quadratic time.  
This is an improvement over the existing predictability algorithms 
that run in $O(n^4)$.  

This paper is organised as follows.  
Next section presents preliminary definitions.  
Our definition of predictability is presented in Section~\ref{sec::defi}, 
together with a discussion of its benefits.  
Our algorithm is given in Section~\ref{sec::solution}.  
Existing approaches are discussed in Section~\ref{sec::related}.  

\section{Preliminaries}
\label{sec::prelim}
\subsection{Discrete Event Systems}

This work is applicable to finite discrete event systems (DES) 
\cite{cassandras-lafortune::99}.  
The system is modeled as a DES and is assumed fixed for this paper.  
A (finite) DES is a model for dynamic systems 
where the state space is discrete (and finite) 
and is modeled as a finite state machine.

A (partially observable) \emph{finite state machine} (FSM) 
is a tuple $A = \langle Q,\Sigma,T,\qi,\sigmao\rangle$ 
where $Q$ is a finite set of states, 
$\Sigma$ is a finite set of events, 
$T \subseteq Q \times \Sigma \times Q$ is a finite set of transitions, 
$\qi \in Q$ is the initial state, 
and 
$\sigmao \subseteq \Sigma$ is a finite set of observable events.  

To simplify notations, it is assumed that the FSM is deterministic, 
i.e., there is only one initial state 
and there are no two transitions originating from the same state 
and labeled with the same event: 
\begin{displaymath}
  \left\{\langle q,e,q'_1\rangle,\langle q,e,q'_2\rangle\right\} 
  \subseteq T\ \Rightarrow\ q'_1 = q'_2.  
\end{displaymath}
This assumption is not restrictive 
as any non-deterministic FSM can be turned into a deterministic FSM 
that is equivalent from a predictive/monitoring perspective, 
by adding a number of states and transitions 
smaller than the original number of transitions 
and without affecting the overall complexity of the algorithm.  
Furthermore the algorithms presented later apply 
to non-deterministic FSM as well.  
The assumption of determinism is however convenient 
because there a one-to-one mapping between a path and a trace 
(defined below).  

A \emph{path} $\path$ is a double sequence 
of states and events $q_0\trans{e_1}\dots\trans{e_k}q_k$ 
such that $\forall i \in \{1,\dots,k\},\ 
\langle q_{i-1},e_i,q_i\rangle \in T$.  
The label $u$, called the \emph{trace}, of the path 
is the sequence of events $u = e_1\dots e_k$.  
That there exists a path labeled by $u$ from $q_0$ to $q_k$ 
is denoted $q_0 \trans{u} q_k$; 
the state $q_k$ reached from $q$ through $u$ is denoted $q \trans{u}$ 
and the fact that it exists is written $(q \trans{u})\ \in Q$.  

The definition of a path is extended to infinite paths 
$q_0\trans{e_1}q_1\trans{e_2}\dots$ such that for all $i\ge 0$, 
$q_0\trans{e_1}\dots\trans{e_i}q_i$ is a path.  
It is assumed that the system is live, 
i.e., that for any state $q \in Q$, there exists an outgoing transition: 
$\forall q \in Q,\ \exists e \in \Sigma.\ \exists q' \in Q.\ 
\langle q,e,q'\rangle \in T$.  
Infinite traces are denoted $w$ and finite ones $u$.  
The \emph{prefix} relation is denoted $u \prefix v$ 
where $v$ may be finite or infinite.  
We extend the notation $(q \trans{w})\ \in Q$ to infinite traces, 
with the meaning $\forall u \prefix w.\ (q \trans{u}) \in Q$.  

The system starts in state $q_0 = \qi$ 
and takes an infinite path.  
The language $\lang = \{w \in \Sigma^\omega \mid (\qi \trans{w})\ \in Q\}$ 
is defined as the set of infinite words over $\Sigma$ 
that label an infinite path on the FSM starting from the initial state.  

Given a finite word $u \in \Sigma^*$, 
the \emph{observation} of $u$ is the traditional projection 
of $u$ on the set of observable events: 
\begin{displaymath}
\obs(u) = \left\{
\begin{array}{l l}
\varepsilon & \textrm{ if } u = \varepsilon, \\
\obs(u')    & 
\textrm{ if } u = eu' \textrm{ and } e \in \Sigma \setminus \Sigmao, \\
e\ \obs(u')  & 
\textrm{ if } u = eu' \textrm{ and } e \in \Sigmao 
\end{array}
\right.
\end{displaymath}
where $\varepsilon$ is the empty sequence.  
As usual it is assumed that any infinite trace 
generates infinitely many observations.  

\subsection{Faults}

The system can be subject to faults, 
i.e., types of behaviour that we wish to prevent.  
Faults can be defined as a single event 
or as a subtle pattern of events \cite{jeron-etal::dx::06}.  
These two definitions are however very similar: 
the important notion here 
is that it can also be modeled as the property 
of the current (possibly augmented) state of the system 
(normal state vs. faulty state).  
A set $F \subseteq Q$ of states will represent the faulty states: 
a path is faulty if it reaches a faulty state 
($\exists i.\ q_i \in F$).  
The faulty aspect of a trace $u$ 
will therefore be represented by $(\qi \trans{u}) \in F$.  
Notice that, by definition, any transition from a faulty state 
leads to a faulty state: 
\begin{displaymath}
  \langle q,e,q'\rangle \in T \ \land\ q \in F \Rightarrow q' \in F.  
\end{displaymath}
It is assumed that the initial state is not faulty.  
The set of infinite faulty traces is represented 
by language $\langf{F} \subset \lang$, 
which is formally defined as the set of traces 
whose path from $\qi$ is faulty.  

\section{$(i,j)$-Predictability}
\label{sec::defi}
\subsection{Predictability}

Fault prediction is the problem of deciding 
whether an operator should be warned that a fault is bound to occur.  
We want to give guarantees about the prediction of the fault.  
This guarantee is expressed by a tuple $(i,j)$ 
where $i$ (resp. $j$) is a lower bound (resp. upper bound) 
of the fault occurrence.  

In the following a \emph{time interval} is a pair of elements $(x,y)$ 
from $\naturals \cup \set{\infty}$ 
(the natural numbers including zero and infinity) 
so that $x \le y$.  
We define the operator $\ominus$ 
so that $(x,y) \ominus 1 = (x \ominus 1,y \ominus 1)$ 
where $\ell \ominus 1 = \ell$ if $\ell \in \set{0,\infty}$ 
and $\ell \ominus 1 = \ell - 1$ otherwise.  
A time interval $(x,y)$ can be interpreted 
as the set of numbers between $x$ and $y$.  
Under this interpretation 
the relation $(x,y) \subseteq (x',y')$ 
is equivalent to $x' \le x \le y \le y'$; 
and $(x,y) \cup (x',y') = (\min(x,x'),\max(y,y'))$.%
\footnote{Notice that $(x,y) \cup (x',y')$ may contain elements 
that are neither in $(x,y)$ nor in $(x',y')$.}

A predictor is a machine $P$ that, given a sequence $o$ of observations, 
returns a time interval $(x,y) = P(o)$, 
meaning that any trace that matches this sequence 
will not become faulty before $x$ more observations are collected 
(if $x = 0$, the fault may already have occurred) 
but will definitely be faulty before $y$ more observations are 
(or returns $y = \infty$ if the fault is not predicted---%
it may never occur).  
In the coming definition, 
notice that, while this is not explicitely stated, 
if $u$ and $u'$ are two different traces 
that generate the same observations ($\obs(u) = \obs(u')$) 
then the predictor should obviously give the same prediction: 
$P(\obs(u)) = P(\obs(u'))$.  
Hence the predictor has to be conservative 
so as to satisfy the two constraints given in the definition 
for all relevant traces.  
In other words, there are two types of uncertainty: 
uncertainty about what happened until now 
(we only know that the behaviour generated the sequence $o$ 
but the actual behaviour is unknown); 
uncertainty about what will happen from now.  

\begin{defi}\label{defi::predictor}
  A \emph{predictor} is a machine $P$ 
  that takes a sequence of observations 
  and that returns a time interval 
  with the following property: 
  $\forall w \in \lang.\ \forall u_1,u_2$ 
  such that $u_1 \prefix u_2 \prefix w$, 
  let $(x,y) = P(\obs(u_1))$, then 
  \begin{itemize}
  \item 
    $|\obs(u_2)| - |\obs(u_1)| < x 
    \Rightarrow (\qi \trans{u_2})\ \not\in F$ and 
  \item 
    $|\obs(u_2)| - |\obs(u_1)| \ge y 
    \Rightarrow (\qi \trans{u_2})\ \in F$.  
  \end{itemize}
\end{defi}

An $(i,j)$-predictor has the added requirement 
that, before a fault occurs, 
a prediction should be made about the fault occurrence 
that is tighter than, or as tight as, $(i,j)$.  

\begin{defi}\label{defi::ijpredictor}
  A predictor $P$ is an \emph{$(i,j)$-predictor
  for a given trace $w \in \langf{F}$} if 
  \begin{displaymath}
    \exists u \prefix w.\ P(\obs(u)) \subseteq (i,j).  
  \end{displaymath}

  A predictor is an \emph{$(i,j)$-predictor} 
  if it is an $(i,j)$-predictor for every trace $w \in \langf{F}$.  
\end{defi}

$(i,j)$-predictability is then the property 
that an $(i,j)$-predictor exists.  
We also define \emph{$i$-predictability}, 
the property that the fault occurrence can be predicted 
at least $i$ observations before it occurs; 
and \emph{predictability}, 
the property that the fault can be predicted before it occurs.  

\begin{defi}\label{defi::predictability}
  A system is \emph{$(i,j)$-predictable} 
  if there exists an $(i,j)$-predictor for it.  
  It is \emph{$i$-predictable} if it is $(i,j)$-predictable 
  for some $j \in \naturals$.  
  It is \emph{predictable} if it is $i$-predictable 
  for some $i \in \naturals \setminus \set{0}$.  
\end{defi}

Notice that the condition $j \in \naturals$ (i.e., $j \neq \infty$) 
is necessary because forbidding the upper bound of $P(o)$ 
to be $\infty$ forces the predictor 
to predict the fault before its occurrence 
(i.e., the predictor asserts that the fault will definitely occur).  
Similarly we forbid $i = 0$ because we want the fault to be predicted 
in a state where it has not occurred yet.  

\begin{figure}[ht]
  \begin{center}
\begin{tikzpicture}[->,>=stealth',shorten >=1pt,auto,node distance=2cm,
                    semithick]
  \tikzstyle{every state}=[fill=none,draw=black,text=black]

  \node[state,initial,initial text=] (A) at (0,2) {};
  \node[state] (B) at (2,2) {};
  \node[state] (C) at (0,0) {};
  \node[state] (D) at (2,0) {};
  \node[state] (E) at (4,2) {};
  \node[state] (F) at (4,0) {};
  \node[state,fill=gray] (G) at (6,0) {};

  \path
    (A) edge [bend left] node {$a$} (B) 
    (B) edge [bend left] node {$b$} (A)
    (C) edge [bend left] node {$a$} (D) 
    (D) edge [bend left] node {$c$} (C)
    (A) edge node {$t$} (C)
    (B) edge node {$d$} (E)
    (D) edge node {$d$} (F)
    (E) edge node {$c$} (F)
    (F) edge node {$a$} (G)
    (G) edge [loop] node {$a$} (G)
  ;
\end{tikzpicture}
  \end{center}
  \caption{Example of a system; $t$ is the only unobservable event.}
  \label{fig::simple}
\end{figure}
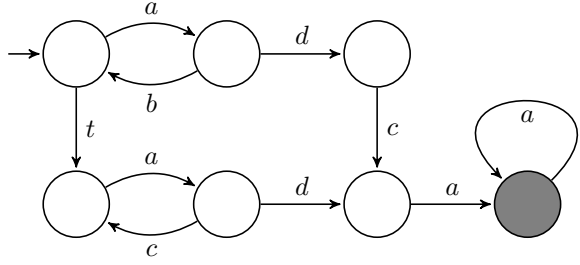

\begin{table}[ht]
  \begin{center}
    \begin{tabular}{c | c}
      Observation pattern & Prediction\\
      \hline
      No $d$ & $[2,\infty]$\\
      Last observed event is $d$ & $[1,2]$\\
      Second last observed event is $d$ & $[0,1]$\\
      Contains $d$ followed by two or more observed events & $[0,0]$
    \end{tabular}
  \end{center}
  \caption{A $(1,2)$-predictor for the system of Figure~\ref{fig::simple}.}
  \label{tab::predictor}
\end{table}

These definitions are illustrated 
with the example of Figure~\ref{fig::simple}.  
The faulty states are represented with grey filling.  
Table~\ref{tab::predictor} presents one predictor.  
For instance the first pattern of the predictor specifies 
that if the sequence of observations does not contain the event $c$ 
then the prediction is $(2,\infty)$, 
i.e., there will be at least two observations 
before the fault occurs, and it may never occur.  
The second pattern specifies 
that if the last event of the sequence of observations is $d$ 
then the prediction is $(1,2)$, meaning that a faulty state 
will be reached after one or two more observations are received.  
Similarly for the third pattern: the prediction is $(0,1)$, 
i.e., it may already have occurred or it will 
when the next observation has been received.  
Finally the last pattern indicates a situation 
where the fault definitely occurred.  

We illustrate that the machine in Table~\ref{tab::predictor} 
(denoted $P$ here)
indeed presents a predictor on a few selected examples.  
We first assume a trace $u_1 = aba$ 
with prediction $P(\obs(u_1)) = (2,\infty)$.  
Consider its continuation $u_2 = u_1b$; 
then the length difference between $\obs(u_2)$ and $\obs(u_1)$ is $1$, 
which is less than $2$; 
therefore $u_2$ has to satisfy $(\qi \trans{u_2}) \not\in F$, 
which it does.  
Consider instead $u_2 = u_1dc$; 
the length difference is this time $2$, 
which means that none of the constraints 
in Definition~\ref{defi::predictor} applies.  
Predictor $P$ is not claimed to be ``optimal'' 
(where the precise definition of optimality is presented later); 
nevertheless one might claim 
that a prediction of $(2,\infty)$ is not very precise 
given that any continuation of $u_1$ requires three observable events 
to reach a faulty state ($dca$ is the shortest).  
Notice however that $P$ does not know that the system trace is $u_1$: 
it only knows the sequence of observations generated by $u_1$, 
i.e., $aba$, which is identical to the sequence 
generated by $u'_1 = abta$; 
this trace $u'_1$ can reach a faulty state in just two observable steps 
($da$), which forces the lower bound of $P(\obs(u_1))$ to be at most $2$.  

Assume now $u_1 = abad$ with prediction $(1,2)$.  
Consider the non-faulty trace $u_2 = u_1c$; 
the length difference is $1$, 
which means that none of the constraints 
in Definition~\ref{defi::predictor} applies.  
Consider instead the faulty trace $u_2 = u_1ca$; 
the length difference is $2$, 
which is greater or equals to the upper bound of the prediction; 
therefore $u_2$ has to satisfy $(\qi \trans{u_2}) \in F$, 
which it does.  

As we can see any faulty trace has to include $d$, 
which means that the flow of observations generated by a faulty trace 
will eventually be associated with the prediction $(1,2)$.  
Therefore the system is $(1,2)$-predictable.  
We can however show that the system is not $(2,2)$-predictable.  
Indeed consider the infinite faulty trace $w = adca^\omega$ 
where the exponent $^\omega$ indicates an infinite repetition of $a$.  
For $w$ to be $(2,2)$-predictable, 
we need to exhibit one of its prefix $u_1$ 
such that one can predict $P'(\obs(u_1)) \subseteq (2,2)$ 
(here $P'(\obs(u_1))$ should exactly equal $(2,2)$).  
Assume that such a prefix and such a predictor exist.  
Following Definition~\ref{defi::predictor}, 
consider a continuation $u_2$ of $u_1$ 
that generates one more observation; 
because $|\obs(u_2)| - |\obs(u_1)| = 1$, 
$u_2$ should not lead to a faulty state.  
Therefore $u_1$ has to belong to the set 
$\set{\varepsilon, a, ad}$.  
Similarly however, if $u_2$ is chosen such that its observable length 
is exactly two more than that of $u_1$, 
then $u_2$ has to lead to a faulty state.  
Therefore $u_1 = ad$ and $P'(\obs(u_1)) = P'(ad) = (2,2)$.  
Consider however the trace $u'_1 = tad$ 
and its continuation $u'_2 = u'_1 a$.  
Clearly $P'(\obs(u'_1)) = P'(ad) = (2,2)$.  
According to Definition~\ref{defi::predictor} 
since $|\obs(u'_2)| - |\obs(u'_1)| = 1 < 2$ 
$u'_2$ should not lead to a faulty state.  
It does however, which shows that no prefix $u_1$ of $w$ 
satisfies $P'(u_1) \subseteq (2,2)$ for some predictor $P'$.  

\subsection{Discussion}

Predictors can be used to stop or rectify the system 
before it produces a faulty behaviour.  
Being able to predict a fault well in advance 
helps getting prepared for intervention; 
this is represented by the $i$ parameter (which should be maximised).  
Being able to predict the time when the fault is likely to happen 
prevents hasty corrections; 
this is represented by the difference $(j-i)$ 
(which should be minimised).  
There is an implicit assumption here 
that the number of observations is indicative of time: 
for instance the system generates one observation per minute.  
This is particularly relevant to hybrid systems modeled as DES 
\cite{vento-etal::tsmc::15}.  

Ideally the system should be $(i,j)$-predictable with a large $i$ value 
and a small $(j-i)$ value.  

We illustrate the definition of predictability by considering the example 
of the potentially critical subsystem of an aircraft.  
This example is, of course, very limited.  
For such a system it is important to predict faults well in advance 
in order to take preventive measures 
(e.g., modify the flight path in order to stay near to an aerodrome).  
On the other hand it is also important to provide a precise prediction 
as emergency landings are expansive.  

In order to provide an early prediction 
we might want the system to be at least $30$-predictable.  
At that stage however, we do not need a precise prediction: 
a $(30,10\,000+)$-predictability is still acceptable.  
For the second requirement however, 
we want to be able to predict the fault quite accurately, 
for instance $(15,240)$-predictability
which suggests that the fault will occur in the next four hours 
and that an unscheduled landing is now necessary.  
So, interestingly, this example requires 
two different predictability properties.  

\section{Solving Interval Predictability Problems}
\label{sec::solution}
This section shows how to verify the predictive level 
of a given system.  

\subsection{Predictive levels}

We first show that, 
while the definition of predictability involves two parameters, 
the dimension of predictability is actually much smaller.  

\begin{lemm}\label{lemm::pred->pred}
  A system that is $(i,j)$-predictable 
  is also 
  \begin{enumerate}
  \item 
    $(i,(j+1))$-predictable (if $j \neq \infty$) 
    and
  \item 
    $((i-1),(j\ominus 1))$-predictable (if $i \ge 2$). 
  \end{enumerate}
\end{lemm}

\begin{proof}
  That $(i,j)$-predictability entails $(i,j+1)$-predictability 
  is trivial from Definition~\ref{defi::ijpredictor}: 
  an $(i,j)$-predictor is also an $(i,j+1)$-predictor 
  since the constraint on the prediction is strictly weaker.  

  Assume that the system is $(i,j)$-predictable with $i \ge 2$, 
  i.e., there exists an $(i,j)$-predictor $P$.  
  Then define $P'$ such that 
  \begin{itemize}
  \item $P'(\varepsilon) = P(\varepsilon)$ and 
  \item $P'(oe) = P(o) \ominus 1$.  
  \end{itemize}  
  It is easy to show that $P'$ is a predictor 
  (if the prediction $P(o)$ was correct, 
  then the prediction $P'(oe)$ is correct).  
  Furthermore it is easy to prove that $P'$ is an $(i-1),(j-1)$-predictor: 
  if $P(\obs(u)) \subseteq (i,j)$ for some prefix $u$ of $w$, 
  then for the prefix $ue$, 
  $P'(ue) = P(u) \ominus 1 \subseteq (i-1,j-1)$ 
  (or $(i-1,\infty)$ if $j = \infty$).  
\end{proof}

Lemma~\ref{lemm::pred->pred} shows that some levels of predictability 
are strictly weaker than others.  
There are however levels of predictability 
that are mutually incomparable.  
Consider the examples of Figure~\ref{fig::1123}.  
Clearly the system of Figure~\ref{fig::1123}a is $(1,1)$-predictable 
because a fault is always preceded by two $a$s 
and the occurrence of the first $a$ implies 
that the fault will be reached after the next observation; 
on the other hand it is not $(2,3)$-predictable 
because when the fault becomes unavoidable 
(i.e., it will occur after less than $3$ observations) 
then the fault can (and, actually, will) occur after less than $2$ observations.  
The system of Figure~\ref{fig::1123}b is $(2,3)$-predictable 
because the fault is always preceded by $aaa$ or $aabb$ 
and because observing a first $a$ implies that the fault is unavoidable; 
on the other hand, it is not $(1,1)$-predictable 
because, after observing $aa$, it is not possible to decide 
whether the fault will occur immediately or after two observations.  

\begin{figure*}[ht]
  \begin{minipage}{0.35\linewidth}
    \begin{center}
\begin{tikzpicture}[->,>=stealth',shorten >=1pt,auto,node distance=2cm,
                    semithick]
  \tikzstyle{every state}=[fill=none,draw=black,text=black]

  \node[state,initial,initial text=] (A) at (0,0) {};
  \node[state] (B) at (2,0) {};
  \node[state,fill=gray] (C) at (4,0) {};

  \path 
    (A) edge        node {$a$} (B)
    (B) edge        node {$a$} (C)
    (A) edge [loop] node {$b$} ()
    (C) edge [loop] node {$a$} ()
  ;
\end{tikzpicture}
    \end{center}
  \end{minipage}\hfill
  \begin{minipage}{0.55\linewidth}
    \begin{center}
\begin{tikzpicture}[->,>=stealth',shorten >=1pt,auto,node distance=2cm,
                    semithick]
  \tikzstyle{every state}=[fill=none,draw=black,text=black]

  \node[state,initial,initial text=] (A) at (0,0) {};
  \node[state] (B) at (2,0) {};
  \node[state] (C) at (4,0) {};
  \node[state] (D) at (6,1) {};
  \node[state,fill=gray] (E) at (6,-1) {};

  \path 
    (A) edge        node {$a$} (B)
    (B) edge        node {$a$} (C)
    (C) edge        node {$a$} (E)
    (C) edge        node {$b$} (D)
    (D) edge        node {$b$} (E)
    (A) edge [loop] node {$b$} ()
    (E) edge [in=-15,out=15,loop] node {$a$} ()
  ;
\end{tikzpicture}
    \end{center}
  \end{minipage}
  \begin{minipage}{0.4\linewidth}
    a. Example of a $(1,1)$-predictable set of faults 
    which is not $(2,3)$-predictable.
  \end{minipage}\hfill
  \begin{minipage}{0.5\linewidth}
    b. Example of a $(2,3)$-predictable (and $(1,2)$-predictable) 
    set of faults 
    which is not $(1,1)$-predictable.
  \end{minipage}
  \caption{Illustrating that some predictive levels are not comparable.}
    \label{fig::1123}
\end{figure*}
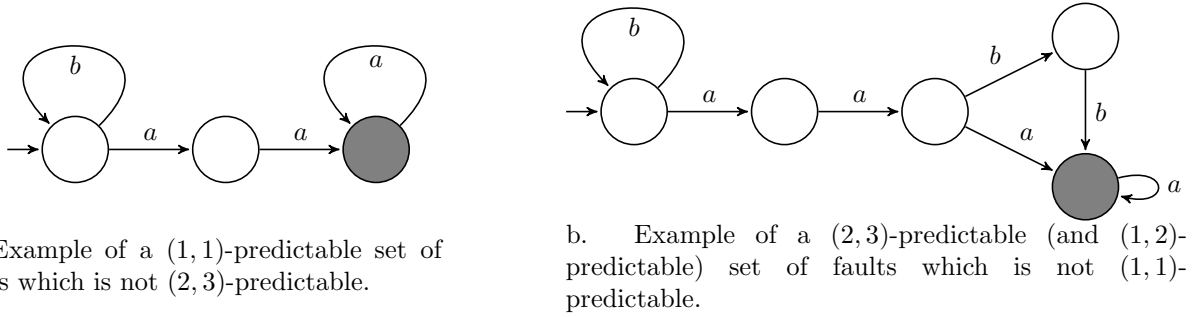

\subsection{Characterisation of Predictability}

In order to determine whether a system is predictable 
we define notions of distance between a system state and a fault.  

\begin{defi}
  The \emph{minimal distance} between $q$ and the set $F$ of states 
  denoted $\dmin{F}(q)$, is the minimum number of observations 
  before reaching $F$ from $q$ 
  \begin{displaymath}
    \dmin{F}(q) = \min_{(q \trans{u}) \ \in F} |\obs(u)| 
  \end{displaymath}
  and $\infty$ if there is no such $u$.  
  The \emph{maximal distance} between $q$ and a set of states $F$, 
  denoted $\dmax{F}(q)$, is the maximum number of observations 
  before reaching $F$ from $q$ 
  \begin{displaymath}
    \dmax{F}(q) = \max_{(q \trans{u}) \ \in Q \setminus F} |\obs(u)| + 1, 
  \end{displaymath}
  $\infty$ if there is no bound to $|\obs(u)|$, 
  and $0$ if $q \in F$ (i.e., there is no such $u$).  
\end{defi}

Notice that these distances are bounded by the number $|Q|$ of states 
when they are different from $\infty$.  
Indeed, if $\dmin{F}(q) \ge |Q|$ 
the corresponding trace includes a cycle, 
and a smaller trace therefore exists (by cutting the cycle).  
Similarly if $\dmax{F}(q) \ge |Q|$ 
the corresponding trace includes a cycle, 
and a longer trace exists (where the cycle can be taken once more).

The minimal and maximal distances give us 
a first estimate of the time interval before fault.  
To simplify notations we write $\distances(q)$ 
to denote the time interval $(\dmin{F}(q),\dmax{F}(q))$.  

\begin{lemm}\label{lemm::singlestateprediction}
  For all trace $w \in \lang$ and all prefix $u \prefix w$, 
  if $P$ is a predictor then
  \begin{displaymath}
    \distances(\qi \trans{u})\ \subseteq\ P(\obs(u)).  
  \end{displaymath}
\end{lemm}

\begin{proof}
  Let $(i,j) = \distances(\qi \trans{u})$ be the time interval 
  of the state $(\qi \trans{u})$ 
  and let $(x,y) = P(\obs(u))$ be the prediction of $\obs(u)$.  

  By definition of the minimal distance~$i$, 
  there exists a trace $u'$ 
  such that $u \prefix u' \prefix w$, 
  $|\obs(u')| - |\obs(u)| = i$, 
  and $(\qi \trans{u'}) \in F$.  
  Therefore $i < x$ would contradict the first condition 
  in the definition of a predictor (Def.~\ref{defi::predictor}).  

  Furthermore if $j > 0$, 
  then by definition of the maximal distance~$j$, 
  there exists a trace $u'$ 
  such that $u \prefix u' \prefix w$, 
  $|\obs(u')| - |\obs(u)| = j-1$, 
  and $(\qi \trans{u'}) \in Q \setminus F$.  
  Therefore $j-1 \ge y$ would contradict the second condition 
  in the definition of a predictor (Def.~\ref{defi::predictor}).  

  If $j = 0$ then $i = 0$ and $y \ge x \ge i = j$ implies $y \ge j$.  
\end{proof}

This result can be generalised to the collection of states 
that an observer can assume the system to be in 
(the \emph{belief state}).  
Formally the belief state $\belief(o)$ is the set of states 
that the system can be in if the sequence $o$ of observations 
has been observed: 
\begin{displaymath}
  \belief(o) = \set{ q \in Q \mid 
  \exists u.\ 
  \qi \trans{u} q\ \land\ \obs(u) = o}.  
\end{displaymath}

\begin{coro}\label{coro::predictorsupset}
  For all predictor $P$, 
  for all sequence $o$ of observations 
  \begin{displaymath}
    \left(\bigcup_{q \in \belief(o)} \distances(q) \right)
    \subseteq P(o).  
  \end{displaymath}
\end{coro}

\begin{proof}
  If $q \in \belief(o)$ is an element of the belief state 
  then, by definition of the belief state, 
  there exists a trace $u$ such that $\qi \trans{u} q$ 
  and $o = \obs(u)$.  
  From Lemma~\ref{lemm::singlestateprediction}, 
  $\distances(q) \subseteq P(o)$, 
  which also applies to the union of these elements.  
\end{proof}

Actually it is possible to characterise 
the ``optimal'' predictor in terms of distances.  
Let $P$ and $P'$ be two predictors.  
We say that $P$ is stronger than $P'$, denoted $P \succeq P'$, 
iff $P(o) \subseteq P'(o)$ for all $o$.%
\footnote{We assume that $P(o)$ and $P'(o)$ 
are undefined if $o$ cannot be generated by the system 
($\belief(o) = \emptyset$).}
We denote $\optipred$ the \emph{optimal predictor}: 
$\optipred = \max_\succeq \set{P \mid P \textit{ is a predictor}}$.  
It should be clear that the optimal predictor is well-defined and unique.  

\begin{lemm}\label{lemm::defopti}
  The optimal predictor $\optipred$ 
  is exactly the predictor that satisfies $\optipred(o) = 
  \left(\bigcup_{q \in \belief(o)} \distances(q) \right)$ 
  for all sequence $o$ of observations.  
\end{lemm}

\begin{proof}
  Let $(i,j) = \left(\bigcup_{q \in \belief(o)} \distances(q) \right)$ 
  and $(x,y) = \optipred(o)$.  
  From Corollary~\ref{coro::predictorsupset} 
  we already know that $(i,j) \subseteq (x,y)$.  
  We only need to prove that $P(o) = (i,j)$ 
  is a correct prediction.  

  Following Definition~\ref{defi::predictor} 
  let $w \in \lang$ be an infinite trace 
  and let $u_1,u_2$ be two finite traces 
  such that $u_1 \prefix u_2 \prefix w$ and $\obs(u_1) = o$.  
  Let us call $q_1$ the state reached by $u_1$ 
  and $q_2$ the state reached by $u_2$: 
  $\qi \trans{u_\ell} q_\ell$.  
  By definition of the belief state, $q_1 \in \belief(o)$.  
  To prove that $P(o)$ is a correct prediction 
  we need to prove that the two conditions 
  of Definition~\ref{defi::predictor} are satisfied.  

  Assume that $q_2 \not\in F$; 
  we shall prove that the premise of the second condition 
  in Definition~\ref{defi::predictor} is not satisfied.  
  By definition of the maximal distance of $q_1$: 
  $\dmax{F}(q_1) > |\obs(u_2)| - |\obs(u_1)|$.  
  Since we know $j \ge \dmax{F}(q_1)$, it clearly holds 
  that $|\obs(u_2)| - |\obs(u_1)| < j$.  

  Assume instead that $q_2 \in F$; 
  we shall prove this time 
  that the premise of the first condition is not satisfied.  
  By definition of the minimal distance of $q_1$: 
  $\dmin{F}(q_1) \le |\obs(u_2)| - |\obs(u_1)|$.  
  Since we know $i \le \dmin{F}(q_1)$, it clearly holds 
  that $|\obs(u_2)| - |\obs(u_1)| \ge i$.  
\end{proof}

As it turns out $\optipred(o)$ equals the union of exactly two intervals.  

\begin{lemm}\label{lemm::twostates}
  For all sequence $o$ of observations 
  such that $\belief(o) \neq \emptyset$, 
  there exists a pair of states $\set{q_1,q_2} \subseteq \belief(o)$ 
  such that $\optipred(o) = \distances(q_1) \cup \distances(q_2)$.  
\end{lemm}

\begin{proof}
  From Lemma~\ref{lemm::defopti} 
  $\optipred(o)$ is the union of a finite collection of intervals.  
  Because this set is finite, there is an interval, say $\distances(q_1)$, 
  whose lower bound is minimal; 
  similarly there is an interval, say $\distances(q_2)$, 
  whose upper bound is maximal.  
  Therefore $\optipred(o) = \distances(q_1) \cup \distances(q_2)$.  
\end{proof}

The optimal predictor exhibits some very interesting properties.  

\begin{lemm}\label{lemm::onemoreobs}
  For all sequence $o$ of observations, 
  \begin{displaymath}
    \optipred(oe) \subseteq \optipred(o) \ominus 1.  
  \end{displaymath}
\end{lemm}

\begin{proof}
  Let $u_1 \prefix u_2$ be two finite traces 
  such that $|\obs(u_2)| = |\obs(u_1)| + 1$.  
  Then by definition 
  $\dmin{F}(\qi \trans{u_1}) \ge \dmin{F}(\qi \trans{u_2}) + 1$ 
  (unless $\dmin{F}(\qi \trans{u_1}) = 0$).  
  Similarly 
  $\dmax{F}(\qi \trans{u_1}) \le \dmax{F}(\qi \trans{u_2}) + 1$ 
  (unless $\dmax{F}(\qi \trans{u_1}) = \infty$).  
  
  Therefore $\distances(\qi \trans{u_1}) 
  \subseteq \distances(\qi \trans{u_2}) \ominus 1$.  

  For each state $q_2 \in \belief(oe)$, 
  there exists a state in $q_1 \in \belief(o)$ 
  such that two such traces $u_1 \prefix u_2$ 
  lead respectivement to $q_1$ and $q_2$ 
  (but notice that for some $q_1$, there may be no such $q_2$).  
  Therefore $\optipred(oe) = \bigcup_{q_2 \in \belief(oe)} \distances(q_2) 
  \subseteq \bigcup_{q_1 \in \belief(oe)} \distances(q_1) \ominus~1
  = \optipred(o) \ominus 1$.  
\end{proof}

The optimal predictor can be used to decide predictability.  
Indeed from Definition~\ref{defi::ijpredictor} 
any suboptimal predictor enjoys only a (non-necessarily strict) subset 
of $(i,j)$-predictability qualities of the optimal predictor.  
This is expressed in the following corollary 
where non-predictability is proved 
if $(i,j)$ is a strict subset ($\subset$) 
of some prediction $\optipred(o)$.  

\begin{coro}\label{coro::predictfromoptimal}
  If $\dmin{F}(\qi) \ge i$ 
  the system is not $(i,j)$-predictable 
  iff there exists a sequence $o$ of observations 
  such that $(i,j) \subset \optipred(o)$.  
\end{coro}

\begin{proof}
  We assume $\dmin{F}(\qi) \ge i$.  

  $\Leftarrow$
  Assume that there is no sequence $o$ of observations 
  such that $(i,j) \subset \optipred(o)$.  
  Consider a faulty trace $w$.  
  We shall show that $w$ is $(i,j)$-predictable.  

  Let $u \prefix w$ be a faulty prefix: 
  $(\qi \trans{u}) \in F$.  
  Then by Definition~\ref{defi::predictor} of a predictor, 
  $\optipred(\obs(u)) = (x,y)$ where $x = 0$.  
  Notice also that $\optipred(\obs(\varepsilon)) 
  = (x^\varepsilon,y^\varepsilon)$ where $x^\varepsilon \ge i$.  
  From Lemma~\ref{lemm::onemoreobs} 
  we know that adding one observation to a sequence 
  can reduce the lower bound of the interval returned by $\optipred$ 
  only by $1$.  
  Therefore, since the lower bound is greater than or equal to $i$ 
  for $\varepsilon$ and down to $0$ for $\obs(u)$, 
  there is a prefix $u'$ of $u$ 
  such that $\optipred(u') = (x',y')$ and $x' = i$.  
  But since $(i,j) \not\subset (x',y')$, 
  $(x',y') \subseteq (i,j)$ and the faulty trace $w$ 
  is $(i,j)$-predictable (Def.~\ref{defi::ijpredictor}).

  $\Rightarrow$
  Let $o$ be the sequence of observations 
  such that $(i,j) \subset \optipred(o)$ 
  and let $(x,y)$ be this interval $\optipred(o)$.  
  Notice that $y \ge 1$ since $(i,j)$ is not empty.  

  Assume $y \neq \infty$.  
  From Lemma~\ref{lemm::twostates} 
  and from the definition of the time intervals 
  there exists $u_1 \prefix u_2 \prefix w$ 
  and $u'_1 \prefix u'_2 \prefix w'$ 
  such that 
  \begin{itemize}
  \item 
    $\set{w,w'} \in \lang$, 
  \item 
    $\obs(u_1) = \obs(u'_1) = o$, 
  \item 
    $|\obs(u_2)| - |\obs(u_1)| = x$, 
  \item 
    $(\qi \trans{u_2}) \ \in F$, 
  \item 
    $|\obs(u'_2)| - |\obs(u'_1)| = y-1$, 
  \item 
    $(\qi \trans{u'_2}) \ \in Q \setminus F$.  
  \end{itemize}
  We shall prove by contradiction that $w$ is not $(i,j)$-predictable.  
  
  Assume that $w$ is $(i,j)$-predictable.  
  Then there exists a prefix $u_3$ of $w$ 
  such that $\optipred(u_3) \subseteq (i,j)$.  
  Because of the first condition of Definition~\ref{defi::predictor}, 
  this prefix must be such that $|\obs(u_2)| - |\obs(u_3)| \ge i$, 
  and therefore $|\obs(u_1)| - |\obs(u_3)| \ge 0$.  
  We know that $\optipred(\obs(u_1)) \not\subseteq (i,j)$, 
  therefore $|\obs(u_1)| - |\obs(u_3)| \ge 1$ and $u_3 \prefix u_1$.  

  Because $u_1$ and $u'_1$ generate the same sequence of observations, 
  there exists a prefix $u'_3$ of $u'_1$ (and therefore of $u'_2$) 
  that generates the same sequence of observations as $u_3$.  
  Furthermore, we know that $|\obs(u'_2)| - |\obs(u'_3)| = 
  |\obs(u'_2)| - |\obs(u_3)| > |\obs(u'_2)| - |\obs(u_1)| = 
  |\obs(u'_2)| - |\obs(u'_1)| = y - 1$; 
  that is: $|\obs(u'_2)| - |\obs(u'_3)| \ge y$.  
  According to the second condition of Definition~\ref{defi::predictor}, 
  $(\qi \trans{u'_2}) \ \in F$, 
  which contradicts the last item of the six items 
  presented at the beginning of this proof.  

  The proof under the assumption that $y = \infty$ is very similar.  
  We choose $u'_2$ such that $|\obs(u'_2)| - |\obs(u'_1)| > |Q|+2$.  
  This proves that the system is not $(i,|Q|+1)$-predictable.  
  Since we know that a bound bigger than $|Q|$ is equivalent 
  to that of $\infty$, we show 
  that the system is not $(i,\infty)$-predictable.  
\end{proof}

Notice that if $\dmin{F}(\qi) < i$, 
then the system is not $(i,j)$-predictable for any $j$ 
(even $j = \infty$).  

Combining Corollary~\ref{coro::predictfromoptimal} 
and Lemma~\ref{lemm::twostates}, 
we obtain the following theorem.  

\begin{theo}\label{theo::computingpredictability}
  The system is $(i,j)$-predictable 
  iff $\dmin{F}(\qi) \le i$ and 
  for all sequence $o$ of observations, 
  for all pair of states $(q_1,q_2) \subseteq \belief(o)$, 
  \begin{displaymath}
    (i,j) \not\subset \distances(q_1) \cup \distances(q_2).  
  \end{displaymath}
\end{theo}

We write $q_1 \sim q_2$ the relation 
indicating that the two states $q_1$ and $q_2$ 
appear together in a belief state.  
Notice that $\sim$ is not an equivalence relation 
(it is not transitive).  

\subsection{Algorithms}

We now turn to implementation of Theorem~\ref{theo::computingpredictability}.  
The algorithm includes four steps: 
\begin{enumerate}
\item 
  Compute the minimal distance for each state; 
\item 
  Compute the maximal distance for each state; 
\item 
  Compute the twin plant which represents the $\sim$ relation; 
\item 
  Compute the $(i,j)$-predictability.  
\end{enumerate}
All parts of the verification process will be presented here 
to ensure the paper is self-contained.  

Algorithm~\ref{algo::dmin} computes the minimal distance of each state.  
In this algorithm and the following one, 
$c(e) = 1$ if $e$ is observable and $0$ otherwise.  
It assumes that all states have infinite distance 
until it is has been proved that a shorter distance exists.  
It then sets all faulty states' minimal distance to $0$ 
and updates the minimal distances of all states 
until convergence is reached.  
To make sure that the states are explored in the optimal order 
we use a priority queue $\queue$ 
that orders its elements by smaller value $\dmin{F}(q)$; 
however since $\queue$ only contains elements 
with two types of distances 
(the current distance and this distance plus one), 
the queue can be implemented with two buckets.  
The complexity of the algorithm is therefore linear 
in the number $|T|$ of transitions.  

\begin{algorithm}[ht]
  \begin{algorithmic}
  \STATE 
    {\bf Input}: an FSM $\langle Q,\Sigma,T,\qi,\sigmao\rangle$, 
    a set of states $F \subseteq Q$
  \STATE 
    Create a table $\dmin{F}: Q \rightarrow \mathbf{N} \cup \{\infty\}$
  \FORALL{$q \in Q$}
    \STATE $\dmin{F}(q) := \infty$
  \ENDFOR
  \STATE 
    $\queue = \emptyset$
  \FORALL{$q \in F$}
    \STATE $\dmin{F}(q) := 0$
    \STATE $\queue := \queue \cup \set{q}$
  \ENDFOR
  \WHILE{$\queue \neq \emptyset$} 
    \STATE $q' := \mathrm{pop}(\queue)$
    \FORALL{$\langle q,e,q'\rangle \in T$}
      \IF{$\dmin{F}(q) > \dmin{F}(q') + c(e)$}
        \STATE $\dmin{F}(q) := \dmin{F}(q') + c(e)$
        \STATE $\queue := \queue \cup \set{q}$
      \ENDIF
    \ENDFOR
  \ENDWHILE
  \STATE {\bf return} $\dmin{F}$ 
  \end{algorithmic}
  \caption{Computing the minimal distance.}
  \label{algo::dmin}
\end{algorithm}

\paragraph{}
Algorithm~\ref{algo::dmax} computes the maximal distance of each state.  
It starts by computing the list of states ($N$) 
that can stay outside of $F$ forever 
(those states have infinite maximal distance).  
It then initialises every state with a maximal distance of $0$ 
and updates the distance whenever it finds a bigger value.  
This update will eventually terminate 
(after at most $|Q|$ iterations).  
The first part of the algorithm requires to iterate twice
over all transitions; 
the second part requires to iterate at most $|Q|$ times 
over at most all transitions.  
Therefore the complexity of Algorithm~\ref{algo::dmax} 
is at most $|Q| \times |T|$.  

\begin{algorithm}[p]
\begin{algorithmic}
\STATE 
  {\bf Input}: an FSM $\langle Q,\Sigma,T,\qi,\sigmao\rangle$, 
  a set of states $F \subseteq Q$
\STATE 
  Let $N := Q \setminus F$
\STATE 
  Create map $\nbsucc: N \rightarrow \naturals$
\STATE
  $R := \emptyset$
\FORALL{$q \in N$}
  \STATE $\nbsucc[q] := |\set{ \langle q, e, q'\rangle \in T}|$ 
  \IF{$\nbsucc[q] = 0$}
    \STATE $R := R \cup q$
  \ENDIF
\ENDFOR
\WHILE{$R \neq \emptyset$}
  \STATE 
    Let $q' := \mathrm{pop}(R)$
  \FORALL{$\langle q, e, q'\rangle \in (N \times \Sigma \times \set{q'})$}
    \STATE $\nbsucc[q] := \nbsucc[q] - 1$
    \IF{$\nbsucc[q] = 0$}
      \STATE $R := R \cup q$
    \ENDIF
  \ENDFOR 
\ENDWHILE
\STATE 
  Create a table $\dmax{F}: Q \rightarrow \mathbf{N} \cup \{\infty\}$
\FORALL{$q \in Q$}
  \IF{$q \in N$} 
    \STATE $\dmax{F}(q) := \infty$
  \ELSE
    \STATE $\dmax{F}(q) := 0$
  \ENDIF
\ENDFOR
\STATE 
  $\textit{needsUpdate} := \textit{true}$
\WHILE{$\textit{needsUpdate}$}
  \STATE $\textit{needsUpdate} := \textit{false}$
  \FORALL{$q' \in Q \setminus N$}
    \FORALL{$\langle q, e, q'\rangle \in T$}
      \IF{$\dmax{F}(q) < \dmax{F}(q') + c(e)$}
        \STATE $\dmax{F}(q) := \dmax{F}(q') + c(e)$
        \STATE  $\textit{needsUpdate} := \textit{true}$
      \ENDIF
    \ENDFOR
  \ENDFOR
\ENDWHILE
\STATE {\bf return} $\dmax{F}$ 
\end{algorithmic}
  \caption{Computing the maximal distance.}
  \label{algo::dmax}
\end{algorithm}

\paragraph{}
The twin plant \cite{jiang-etal::tac::01} is a construction 
that determines precisely the $\sim$ relation.  
Notice that, strictly speaking, it is not necessary 
to build it as a finite state machine: 
for predictability only the $\sim$ relation matters; 
not the transitions between the states of the twin plant.  

  Given an FSM $A = \langle Q,\Sigma,T,\qi,\sigmao\rangle$, 
  the \emph{twin plant} is the finite state machine 
  $\langle Q^T,\Sigma^T,T^T,\qi^T,\sigmao\rangle$ 
  where 
  \begin{itemize}
  \item 
    $Q^T = Q \times Q$, 
  \item 
    $\Sigma^T = \big((\Sigma \setminus \sigmao) \times \{1,2\}\big)
    \ \cup\ \Sigmao$, 
  \item 
    $T^T = $

    $
    \begin{array}{c c c @{\ \ \land\ \ } c @{\ \ \land\ \ }c@{}c}
      &
      \set{ 
      \langle\langle q_1,q_2\rangle,e^T,\langle q'_1,q'_2\rangle\rangle\mid&
      \langle q_1, e^T, q'_1\rangle \in T
      &
      \langle q_2, e^T, q'_2\rangle \in T
      &
      e^T \in \sigmao
      &
      }
      \\
      \cup & 
      \set{ 
      \langle\langle q_1,q_2\rangle,e^T,\langle q'_1,q'_2\rangle\rangle\mid&
      \langle q_1, e^T, q'_1\rangle \in T
      &
      q_2 = q'_2 
      &
      e^T \in \Sigma \setminus \sigmao
      &
      }
      \\
      \cup &
      \set{ 
      \langle\langle q_1,q_2\rangle,e^T,\langle q'_1,q'_2\rangle\rangle\mid&
      q_1 = q'_1 
      &
      \langle q_2, e^T, q'_2\rangle \in T
      &
      e^T \in \Sigma \setminus \sigmao
      &
      }.  
    \end{array}$
  \item 
    $\qi^T = \langle \qi,\qi\rangle$.  
  \end{itemize}

It is well-known that the state $\langle q_1,q_2\rangle$ of $Q^T$ 
is reachable from $\qi^T$ iff $q_1 \sim q_2$ 
(Lemma~10, \cite{genc-lafortune::automatica::09} 
where the twin plant is called \emph{verifier}).  
The twin plant can therefore be used to verify the $(i,j)$-predictability.  

The procedure for computing the $(i,j)$-predictability 
is given in Algorithm~\ref{algo::pred}.  
It generates an array $p$ such that for all $i$, 
the system is $(i,p[i])$-predictable and non-$(i,p[i]-1)$-predictable.  
The algorithm first initialises $p[i]$ to $i$.  
It then iterates over all the states of the twin plant 
and updates the table $p$.  
According to Theorem~\ref{theo::computingpredictability} 
the result of Algorithm~\ref{algo::pred} 
is the list of $(i,p[i])$-predictabilities that the system enjoys.  

\begin{algorithm}
  \begin{algorithmic}[1]
    \STATE {\bf input} an FSM $A = \langle Q,\Sigma,T,\qi,\sigmao\rangle$, 
    the list of minimal and maximal distances $\dmin{F}$ and $\dmax{F}$, 
    the twin plant $\langle Q^T,\Sigma^T,T^T,\qi^T,\sigmao\rangle$
    \STATE Create an array of integers $p: \naturals^{|\min{F}(\qi)|}$
    \FORALL{$i \in \{1,\dots,|Q|\}$}
        \STATE $p[i] := i$
    \ENDFOR
    \FORALL{$\langle q,q'\rangle \in \mathrm{Reachable}(Q^T)$}
    \STATE $i := \min (\dmin{F}(q),\dmin{F}(q'))$
    \STATE $j := \max (\dmax{F}(q),\dmax{F}(q'))$
      \STATE $p[i] := \max(p[i],j)$\label{line::ss->notpred}
    \ENDFOR
    \STATE {\bf return} $p$
  \end{algorithmic}
  \caption{Algorithm to compute $(i,j)$-predictability}
  \label{algo::pred}
\end{algorithm}

We claim that the algorithm presented here is quadratic 
in the number of states and transitions of the system.  
It is easy to see that computing the distances 
is at most quadratic for both types of distances, 
and that the resulting structure has linear size 
with constant time access.  
The size of the twin plant is quadratic 
in the size of the original system 
(it includes at most $|Q|^2$ states 
and $\left(2 \times |T| \times |Q|\right) + |T|^2$ transitions)---%
and that is assuming a non-deterministic model. 
Finally the fourth step 
requires iterating over the quadratic number of states in the twin plant.  

Our definitions of predictability and $i$-predictability 
match those of J\'eron et al. \cite{jeron-etal::wc::08} 
and the proposed algorithm can therefore be used 
to verify these properties.  
It is also possible to simplify it 
by focussing on the $i$ parameter.  

As a last result, consider a \emph{fully observable system}, 
i.e., a system in which $\obs(u) = u$.  
Then, at any time, the state of the system can be deduced 
from the sequence of observations; 
but notice that how the system will evolve 
remains unknown.  
Then the relation $\sim$ is equivalent to identity: 
$q \sim q'$ iff $q = q'$.  
Consequently, after the distances of each state have been computed, 
the predictability can be computed in linear time.  

\subsection{Building the Optimal Predictor}

Lemma~\ref{lemm::defopti} gives us a procedure 
for computing the optimal predictor.  
Similarly to diagnosis and its diagnoser \cite{sampath-etal::tac::95} 
it is possible to compute a deterministic FSM 
that represents how the belief state evolves 
as more observations are gathered.  

Formally the \emph{optimal predictor} is a finite state machine 
$\langle \optimal{Q}, \optimal{\Sigma}, \optimal{T}, \optimal{\qi}\rangle$ 
where 
\begin{itemize}
\item 
  $\optimal{Q} = \set{\optimal{q} \mid \optimal{q} \subseteq Q}$, 
\item 
  $\optimal{\Sigma} = \Sigmao$, 
\item 
  $\optimal{T} \subseteq 
  \optimal{Q} \times \optimal{\Sigmao} \times \optimal{Q}$ 
  is defined below, 
  and 
\item 
  $\optimal{\qi} = \set{\qi}$. 
\end{itemize}

For every state $\optimal{q_1} \in \optimal{Q}$ of the optimal predictor
and every event $e \in \optimal{\Sigma}$, 
there is exactly one state $\optimal{q_2}$ 
such that $\optimal{q_1} \trans{e} \optimal{q_2}$ 
is a transition of the optimal predictor.  
The state $\optimal{q_2} \subseteq Q$ is defined 
as the set of states of the system 
that can be reached from a state of $\optimal{q_1}$ 
through a path that generates only one observation: 
\begin{displaymath}
  \optimal{q_2} = 
  \set{ q_2 \in Q \mid \exists q_1 \in \optimal{q_1}.\ 
    \exists u.\ (q_1 \trans{u} q_2) \land (|\obs(u)| = 1)
  }.
\end{displaymath}

Given a sequence $o$ of observations 
the predictor follows the single path labeled by $o$ 
on the predictor and reaches the state $\optimal{q}(o)$ 
(i.e., the state $\optimal{q}(o)$ 
such that $\optimal{\qi} \trans{o} \optimal{q}(o)$).  
The prediction is then $\bigcup_{q \in \optimal{q}(o)} \distances(q)$.%
\footnote{If the model is correct, 
then the state $\optimal{q}(o) = \set{}$ should never be reached 
and the union is therefore well-defined.}
Adding a single observation $e$ to $o$, 
the new prediction can be easily computed 
by getting the state $\optimal{q}(oe)$ that satisfies 
$\optimal{q}(o) \trans{e} \optimal{q}(oe)$.  
Assuming the optimal FSM 
and the interval associated with each state of the predictor 
are precomputed, 
the optimal prediction of a sequence of observations 
is linear in the size of this sequence 
and the incremental optimal prediction is constant time.  
Notice however that, as is the case with the diagnoser 
\cite{rintanen::ijcai::07}, 
the optimal predictor is exponentially large 
in the number of states of the system.  

\section{Related Work}
\label{sec::related}
Predictability as presented in this paper 
was introduced by Genc and Lafortune 
\cite{genc-lafortune::safeprocess::06}.  
Their approach was however only Boolean: 
they addressed the question 
``can the fault be predicted before it occurs?''  
They presented an exponential space algorithm 
using a structure similar to our optimal predictor.  
They also announced the existence of a polytime algorithm, 
similar to the twin plant used for diagnosability 
and formally presented in an extension of their work 
\cite{genc-lafortune::automatica::09}.  

Together with J\'eron and Marchand, 
they proposed an additional improvement to lower the complexity 
down to quadratic \cite{jeron-etal::wc::08}.  
We claim here that their algorithm is not quite quadratic 
(we discuss this question at the end of this section).  
Their approach is very similar to the approach presented 
in the previous section: 
They construct a twin plant and verify predictability 
by checking whether there exists a pair $q_1\sim q_2$ 
such that $\dmin{F}(q_1) = 0$ and $\dmax{F}(q_2) = \infty$.  

Brand\'an Briones and Madalinski presented 
the notions of $lb$-predictability and $ub$-predictability 
\cite{brandanbriones-madalinski::sccc::11}.  
$ub$-predictability is similar to our definition of $i$-predictability 
meaning that the fault is predicted at least $i$ observations 
before the fault occurs.  
$lb$-predictability is the equivalent 
of our property of $(1,j)$-predictability, 
meaning that it is possible to predict the fault occurrence 
before it occurs but when at most $j$ observations 
are still possible before the fault 
(in other words, the fault prediction is not too early).  

\begin{figure}[ht]
  \begin{minipage}{0.47\linewidth}
  \begin{center}
\begin{tikzpicture}[->,>=stealth',shorten >=1pt,auto,node distance=2cm,
                    semithick]
  \tikzstyle{every state}=[fill=none,draw=black,text=black]

  \node[state] (A) at (0,0) {$A$};

  \node[state] (B1) at (2,2) {$B_1$};
  \node (Bm) at (2,0) {$\dots$};
  \node[state] (Bn) at (2,-2) {$B_n$};

  \node[state] (C) at (4,0) {$C$};

  \node[state] (D1) at (6,2) {$D_1$};
  \node (Dm) at (6,0) {$\dots$};
  \node[state] (Dn) at (6,-2) {$D_n$};

  \path
    (A) edge node[above] {$a$} (B1)
    (A) edge node[below] {$a$} (Bn)
    (B1) edge node[above] {$t$} (C)
    (Bn) edge node[below] {$t$} (C)
    (C) edge node[above] {$a$} (D1)
    (C) edge node[below] {$a$} (Dn)
    (D1) edge[in=-15,out=15,loop] node {$a$} (D1)
    (Dn) edge[in=-15,out=15,loop] node {$a$} (Dn)
  ;
\end{tikzpicture}
  \end{center}
  \end{minipage}%
  \begin{minipage}{0.06\linewidth}
  \ 
  \end{minipage}%
  \begin{minipage}{0.47\linewidth}
  \begin{center}
\begin{tikzpicture}[->,>=stealth',shorten >=1pt,auto,node distance=2cm,
                    semithick]
  \tikzstyle{every state}=[fill=none,draw=black,text=black]

  \node[state] (A) at (0,0) {$A$};

  \node[state] (B1) at (2,2) {$B_1$};
  \node (Bm) at (2,0) {$\dots$};
  \node[state] (Bn) at (2,-2) {$B_n$};

  \node[state] (D1) at (6,2) {$D_1$};
  \node (Dm) at (6,0) {$\dots$};
  \node[state] (Dn) at (6,-2) {$D_n$};

  \path
    (A) edge node[above] {$a$} (B1)
    (A) edge node[below] {$a$} (Bn)
    (B1) edge node[above] {$a$} (D1)
    (B1) edge node[below] {$a$} (Dn)
    (Bn) edge node[above] {$a$} (D1)
    (Bn) edge node[below] {$a$} (Dn)
    (D1) edge[in=-15,out=15,loop] node {$a$} (D1)
    (Dn) edge[in=-15,out=15,loop] node {$a$} (Dn)
  ;

\end{tikzpicture}
  \end{center}
  \end{minipage}

  \begin{minipage}{0.47\linewidth}
  \begin{center}
    {a. System.}
  \end{center}
  \end{minipage}%
  \begin{minipage}{0.06\linewidth}
  \ 
  \end{minipage}%
  \begin{minipage}{0.47\linewidth}
  \begin{center}
    {b. $\varepsilon$-reduction of the system.}
  \end{center}
  \end{minipage}
  \caption{DES (a) and its $\varepsilon$-reduction (b).}
  \label{fig::quadra}
\end{figure}
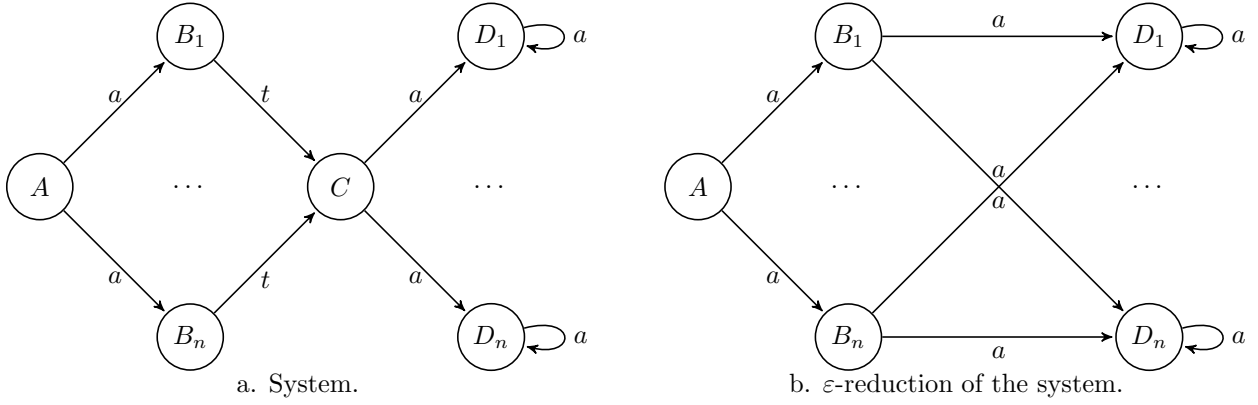

\paragraph{}
While this is a minor issue, 
we provide an example and a comprehensive discussion 
that illustrate the complexity error from J\'eron et al. 
\cite{jeron-etal::wc::08}.  
Consider the example of Figure~\ref{fig::quadra}a.  
This DES includes $2n+2$ states and $4n$ transitions.  
The single observable event is $a$ 
and the single unobservable event is $t$ 
(this example does not feature any faulty event).  
The twin plant then consists in $2n^2+2$ states and $4n^2$ transitions 
(details in Table~\ref{tab::tpsize}).  
The $\varepsilon$-reduction, presented on Figure~\ref{fig::quadra}b, 
contains one state fewer than the original DES 
but $n^2 + 2n$ transitions.  
As a consequence, the number of states in the twin plant 
reduces down to $2n^2+1$ 
but the number of transitions shoots up to $n^4+2n^2$  
(details in Table~\ref{tab::tpsize2}).  

\begin{table}[t!]
  \begin{minipage}{.5\linewidth}
    \begin{center}
      \begin{tabular}{c | c}
        Type of states & Number of states\\ 
        \hline 
        $\langle A, A\rangle$ & $1$\\
        $\langle B_i, B_j\rangle$ & $n^2$\\
        $\langle C, C\rangle$ & $1$\\
        $\langle D_i, D_j\rangle$ & $n^2$\\
        \hline
        Total: & $2n^2 + 2$
      \end{tabular}
    \end{center}
  \end{minipage}%
  \begin{minipage}{.5\linewidth}
    \begin{center}
      \begin{tabular}{c | c}
        Type of transitions & Number of transition\\ 
        \hline 
        $\langle A, A\rangle \rightarrow \langle B_i, B_j\rangle$ & $n^2$\\
        $\langle B_i, B_j\rangle \rightarrow \langle C, C\rangle$ & $n^2$\\
        $\langle C, C\rangle \rightarrow \langle D_i, D_j\rangle$ & $n^2$\\
        $\langle D_i, D_j\rangle \rightarrow \langle D_i, D_j\rangle$ 
        & $n^2$\\
        \hline
        Total: & $4 n^2$
      \end{tabular}
    \end{center}
  \end{minipage}
  \caption{Size of the twin plant for the DES in Figure~\ref{fig::quadra}a.}
  \label{tab::tpsize}
\end{table}

\begin{table}[t!]
  \begin{minipage}{.5\linewidth}
    \begin{center}
      \begin{tabular}{c | c}
        Type of states & Number of states\\ 
        \hline 
        $\langle A, A\rangle$ & $1$\\
        $\langle B_i, B_j\rangle$ & $n^2$\\
        $\langle D_i, D_j\rangle$ & $n^2$\\
        \hline
        Total: & $2n^2 + 1$
      \end{tabular}
    \end{center}
  \end{minipage}%
  \begin{minipage}{.5\linewidth}
    \begin{center}
      \begin{tabular}{c | c}
        Type of transitions & Number of transition\\ 
        \hline 
        $\langle A, A\rangle \rightarrow \langle B_i, B_j\rangle$ & $n^2$\\
        $\langle B_i, B_j\rangle \rightarrow \langle D_k, D_\ell\rangle$ 
        & $n^4$\\
        $\langle D_i, D_j\rangle \rightarrow \langle D_i, D_j\rangle$ 
        & $n^2$\\
        \hline
        Total: & $n^4 + 2 n^2$
      \end{tabular}
    \end{center}
  \end{minipage}
  \caption{Size of the twin plant for the $\varepsilon$-reduced 
  DES in Figure~\ref{fig::quadra}b.}
  \label{tab::tpsize2}
\end{table}

\section{Conclusion}
We presented a notion of $(i,j)$-predictability, 
an extension of predictability that specifies 
that there exists a time interval 
during which the fault occurrence is bound to happen in the system.  
This notion is very useful 
because it allows one to express different type of predictability, 
namely whether a fault can be predicted well in advance, 
whether the time of failure can be precisely predicted, 
or both.  

There are several obvious extensions to these works, 
mainly regarding the expressive power of the modelling framework.  
We want to extend this work to timed systems 
\cite{cassez-grastien::formats::13}, 
to probabilistic systems \cite{nouioua-etal::dx::14}, 
or to hybrid systems \cite{bayoudh-etal::ecai::08}.  
Other works include the extension of the current work 
to decentralised predictors \cite{takai-kumar::tac::12}, 
the study of optimal observability for predictability 
akin to that of diagnosability \cite{brandanbriones-etal::dx::08}
or in combinaison with opacity constraints
\cite{chedor-etal::jdeds::14}.  

\section*{Acknowledgments}
NICTA is funded by the Australian Government 
through the Department of Communications 
and the Australian Research Council 
through the ICT Centre of Excellence Program. 

\bibliographystyle{alpha}
\bibliography{bib}

\end{document}